\documentclass[twoside]{article}

\usepackage{aistats2022}
\input{packages.sty}
\input{macros.sty}

\usepackage[numbers]{natbib}

\begin{document}

%

%

\twocolumn[

\aistatstitle{A Dynamic Programming Algorithm to Compute Joint Distribution of Order Statistics on Graphs}

\aistatsauthor{ Rigel Galgana \And Amy Greenwald \And  Takehiro Oyakawa }

\aistatsaddress{ Brown University \And  Brown University \And Brown University } ]

\begin{abstract}

\noindent Order statistics play a fundamental role in statistical procedures such as risk estimation, outlier detection, and multiple hypothesis testing as well as in the analyses of mechanism design, queues, load balancing, and various other logistical processes involving ranks. 
In some of these cases, it may be desirable to compute the \textit{exact} values from the joint distribution of $\numos$ order statistics. While this problem is already computationally difficult even in the case of $\numrv$ independent random variables, the random variables often have no such independence guarantees.
Existing methods obtain the cumulative distribution indirectly by first computing and then aggregating over the marginal distributions.
In this paper, we provide a more direct, efficient algorithm to compute cumulative joint order statistic distributions of dependent random variables that improves an existing dynamic programming solution via dimensionality reduction techniques. Our solution guarantees a $O(\frac{\numos^{\numos-1}}{\numrv})$ and $O(\numos^{\numos})$ factor of improvement in both time and space complexity respectively over previous methods.

\end{abstract}

\section{Introduction}
\label{sec:introduction}

Let $\rv_1, \ldots, \rv_{\numrv}$ denote a collection of $\numrv$ real-valued random variables.
The $i$th order statistic, denoted as $\rv_{(i)}$, is defined as the $i$th smallest value of $\rv_1, \ldots, \rv_{\numrv}$.

Order statistics often arise in the theory and application of risk estimation and management as they possess strong robustness guarantees considering their ease of interpretation and computational simplicity (\cite{balakrishnan2007permanents, benjamini1995controlling,shorack1996empirical}). They also appear naturally across various disciplines such as mechanism design and auction theory (\cite{varian}), queue inference (\cite{jones1995queues}), and wireless communication and scheduling (\cite{yang2011order}). Their distributional properties, however, are less well understood, especially without the i.i.d. assumption. Various procedures exist for evaluating order statistic distributions at a point, all of which either explicitly or implicitly solve an equivalent combinatorial problem. At a high level, this combinatorial problem involves a partition of the real line into contiguous regions, which we refer to as bins, and tracks the probabilities of the various ways that the random variables could fall within these bins such that certain constraints are satisfied. 
For example, the algorithm described in \cite{boncelet1987algorithms}---which we will be examining more closely in the following section---computes the joint distribution of $\numos$ order statistics $\mathcal{C} = (C_1,\ldots,C_\numos)$ of $n$ random variables by recursively updating the joint probability mass function of the number of random variables inside each bin. While this approach suffers from high computational and memory costs, it can be extended to handle dependent random variables.
Several methods make strict distributional assumptions in order to achieve higher efficiency.  
For example, \cite{glueck2008algorithms} and \cite{schroeder} assume groups of $m \ll n$ homogeneous, independent random variables $\{X_{i, j}\}_{i \in [m], j \in [n_i]}$. \footnote{For notational ease, we let $[a] \equiv \{1,\ldots,a\}$, $[a^*] \equiv \{0,\ldots,a\}, [a^+] \equiv \{1,\ldots,a+1\}.$}.
\cite{galgana2021dynamic} simplified the combinatorial problem structure for this $m$-populations case, allowing polynomial time computation in both $\numos$ and $\numrv$. 

\paragraph{Main Contribution} 
The main contribution of this work is a combinatorial simplification of Boncelet Jr.'s algorithm that significantly reduces its time and space complexity.
With this improvement, we derive an algorithm that offers most of the flexibility of Boncelet Jr.'s and the speed of Galgana and Shi's algorithms. Before describing our approach, we must first formally state the problem of interest. 




\begin{definition}[Galgana and Shi, et. al. (2021)]
    \label{def:jocdf_cts}
    We desire to compute the joint cumulative distribution function (cdf) $F_{\mathcal{C}} (\bm{\rvv})$ of an ordered set $\mathcal{C} = (c_1, \ldots, c_{\numos})$ of $\numos \in \mathbb{Z}_{>0}$ select order statistics of $\numrv \ge \numos$ random variables $X_1,\ldots,X_n \sim \bm{F}$, where $\bm{\rvv} \in \mathbb{R}^{\numos}$. Letting $f_{\mathcal{C}} (\bm{\rvv})$ be the corresponding joint probability density function, we have:%
    \begin{eqnarray}
    \label{eq:jocdf_cts}
    F_{\mathcal{C}}(\bm{\rvv})
        &=& F_{\rv_{(c_1)}, \ldots, \rv_{(c_{\numos})}} (\rvv_{1}, \dots, \rvv_{\numos}) \label{eqn:goal} \nonumber\\ 
        &=& \prob \left( \rv_{(c_1)} \leq \rvv_1, \dots, \rv_{(c_{\numos})} \leq \rvv_{\numos} \right)\\
        &=& \int_{t_1, \ldots, t_{\numos} \in \mathbb{R}^{\numos} \textrm{ s.t.} \, t_i \le \rvv_i, \forall i \in [ \numos ]} f_{\mathcal{C}} (\bm{t}) \, d \bm{t} \nonumber.
    \end{eqnarray}
\end{definition}

As mentioned previously, there are several ways of approaching this problem in various special cases. Our proposed algorithm works best when $\numos$ is small, as the computational complexity 
is still exponential in $\numos$. However our approach is both significantly more computationally and memory efficient than Boncelet Jr.'s whilst retaining the latter's flexibility to handle dependency.

The paper is outlined as follows: In the next section, we focus on the independent, but not necessarily identically distributed random variables case and translate the inherently continuous joint cdf 
problem as in Equation~\eqref{eq:jocdf_cts} to an equivalent combinatorial problem (Section~\ref{sec:combinatorial}) 
. Using this combinatorial setup, the next section details Boncelet Jr.'s algorithm, as well as motivate how Galgana and Shi's dimensionality reduction insight can be used to improve its performance in the case of the joint cdf (Section~\ref{sec:combinatorial}). We then provide rigorous justification to our algorithmic speedup and state the key theorems and recurrence relation. Afterwards, we state the algorithm in its entirety for the i.i.d. case with accompanying complexity analysis and experiments (Section~\ref{sec:solution}). Lastly, we describe the extension to the case of dependent random variables as well as its impact on time and space complexities (Section~\ref{sec:extensions}). We then conclude with a summary of our work (Section~\ref{sec:conclusion}).

\section{The Combinatorial Problem}
\label{sec:combinatorial}

In this section, we borrow heavily from Galgana and Shi's notation as we translate the joint cdf problem of a collection of order statistics of random variables as given in Equation~\eqref{eq:jocdf_cts} to an equivalent combinatorial problem of tossing balls into bins. For simplicity, we assume independence with $X_i \sim F_i$ $\forall i \in [n]$ for the following three sections. Restating this computation as a problem involving balls and bins, we define $\rvv_0 = -\infty$ and $\rvv_{\numos+1} = +\infty$, and partition the real line into $\numos+1$ intervals, $\interval_1, \ldots, \interval_{\numos+1}$, hereafter bins:
%
$(-\infty, +\infty)
    = (\rvv_0, \rvv_{\numos+1}) 
    = \cup_{j=1}^{\numos} (\rvv_{j-1}, \rvv_j] \cup (\rvv_{\numos}, \rvv_{\numos+1})
    = \cup_{j=1}^{\numos+1} \interval_j 
    = I_{1:\numos+1}$.
%
Here, $I_{j:k} = \cup_{i=j}^k \interval_i$ denotes the union of bins $I_j$ through $I_k$ inclusive. Moreover, we let $p_{i, j} = \prob \left( \{ X_i \in I_j \} \right) = \cdf_i (\rvv_j) - \cdf_i (\rvv_{j-1})$ denote the probability that the $i$th ball resides in the $j$th bin, for all $j \in [\numos]$.
For $j, j' \in [\numos]$, we define the key combinatorial objects:
\begin{itemize}

\item Define $\ballcount_{i, j} \doteq \sum_{k=1}^i \mathbf{1}_{X_i \in I_j}$ as the ``number of the first $i$ balls that reside in bin $I_j$.'' 
Overloading notation, we use the random variable $\ballcount_{i, j:j'} = \sum_{k=1}^i \mathbf{1}_{X_i \in I_{j:j'}}$ as ``the number of the first $i$ balls that reside in bins $I_j$ through $I_{j'}$.''

\item Furthermore, we define the random  vectors
$\bm{\ballcount}_i \doteq (\ballcount_{i, 1}, \ldots, \ballcount_{i, \numos}) \in [i]^{\numos}$ to denote ``the number of the first $i$ balls that reside in each bin.''

\item We define the event $\event_{i, j} \doteq \{ \ballcount_{i, 1:j} \geq c_j \}$ as ``the event that at least $c_j$ of the first $i$ balls reside in the first $j$ bins.'' Overloading notation again, we also define the event $\event_{i, j:j'} \doteq \cap_{k=j}^{j'} \event_{i, k}$ as ``the event that at least $c_k$ of the first balls reside in the first $k$ bins, $I_{1:k}$, for all $k \in \{ j, \ldots, j' \}$.''
\end{itemize}

To explain the connection between the original problem of computing $\cdf_{\mathcal{C}} (\bm{x})$ and this combinatorial setup, note that the event $\rv_{(c_j)} \leq \rvv_j$ is means ``the $c_j$th smallest value of the $\numrv$ random variables is less than or equal to $\rvv_j$.'' Equivalently, in the combinatorial setup we have that 
``there are at least $c_j$ balls in bins $I_{1:j}$'': i.e., the event $\event_{n, j}$ holds. We call $\event_{n, j}$ the $j$th bin condition. With this, we revisit Equation~\eqref{eq:jocdf_cts} and see that $\cdf_{\mathcal{C}} (\bm{x})$ is equivalent to:
%
\begin{align*}
\prob ( \bigcap_{j=1}^{\numos} { \rv_{(c_j)} \leq \rvv_j } ) 
= \prob ( \bigcap_{j=1}^{\numos} { \ballcount_{n, 1:j} \geq c_j } ) 
= \prob \left( \event_{n, 1:\numos} \right)
.
\end{align*}
\noindent
In other words, we are interested in computing the probability of satisfying all $\numos$ bin conditions.
\begin{definition}
Assuming 
$\numrv$ independent random variables,
we let $\numrv, \numos, \bm{\rvv}, \mathcal{C},
F_1,\ldots,F_{\numrv}$ be as in Equation~\eqref{eq:jocdf_cts}. With $x_0 = -\infty$, $x_{\numos+1} = \infty$, and $\bm{p} = (p_{i, j})_{i \in [\numrv], j \in [\numos^+]} = (F_i(\rvv_j) - F_i(\rvv_{j-1}))_{i \in [\numrv], j \in [\numos^+]}$, we define the combinatorial problem \ijocdf $(\mathcal{C}, \bm{p})$ as computing the probability $\prob \left( \event_{\numrv, 1:\numos} \right)$.
\end{definition}

This problem statement is fairly general and allows for both continuous and discrete distributions, with the only major assumption being independence. When the independence assumption is relaxed, the definition will change slightly and we will need to construct a corresponding graphical model with corresponding ball throwing order---a schedule. We reserve the details for the later sections, as much of the notation is unnecessary and may detract from the intuitions gained from the solution to the independent variables problem.
\section{Related Works}
\label{sec:related_works}

The primary goal of this paper is to improve the performance of Boncelet Jr.'s algorithm using Galgana and Shi's insights. 
While Boncelet Jr.'s algorithm's performance is independent of the number of underlying distribution, its time and space complexity are exponential in $\numos$ with base $\numrv$ making it viable only in cases where $\numos$ or $\numrv$ are both small. 
While our proposed procedure will remain exponential in $\numos$ like Boncelet Jr.'s algorithm, the base will be smaller. To that end, we first describe both pieces in order to understand their interaction.

\subsection{Boncelet Jr.'s Algorithm}

At a high level, Boncelet Jr.'s algorithm to solve \ijocdf($\mathcal{C}, \bm{p}$) maintains probability tables describing the number of balls in each bin. Using our terminology, the algorithm recursively updates a table of probabilities of $\bm{C}_i = \bm{k}$:

\begin{definition}
    Let $\mathcal{C}$ and $\bm{p}$ be as in \ijocdf($\mathcal{C}, \bm{p}$). We define $e_{j} = (0,\ldots,1,\ldots,0)$ as a length $\numos$ vector of 0's with a 1 in the $j$th position. Then by the law of total probability, we have that $\prob(\bm{C}_i = \bm{k})$ equals: 
    \begin{align*}
        p_{i,\numos+1} 
        \prob ( \bm{C}_{i-1} = \bm{k} )  + \sum_{\{ j \mid \ballcount_j > 0 \}} p_{i,j} 
        \prob ( \bm{C}_{i-1} = \bm{k} - e_j )
        .
    \end{align*}
    The probability of the corresponding \ijocdf is then $\sum_{\bm{k} \in \mathcal{I}} \prob(\bm{C}_n = \bm{k})$, where $\mathcal{I} \doteq \{\bm{i} \in [\numrv*]^{\numos} \mid 0 \leq i_1 \leq \ldots \leq i_{\numos} \leq \numrv, i_j \geq c_j, \forall j \in [\numos] \}$.
\end{definition}

\noindent Using above recurrence, Boncelet Jr.'s algorithm maintains a table $T_i: [\numrv]^{\numos} \to [0, 1]$ of the probabilities of $\prob(\bm{C}_i = \bm{k})$ as a function of table $T_{i-1}$. It then sums over the entries in $T_{n}(\bm{k})$ such that $\bm{C}_{\numrv} = \bm{k}$ satisfies the bin conditions $\event_{n, 1:d}$. This algorithm extends to the dependent random variables setting 
as well, though a modification must be made to track the locations of certain balls required to obtain subsequent conditional distributions. 
In the independent setting, each of the $\numrv$ dynamic programming tables are of size $O(\numrv^{\numos})$ but can be discarded once they are used to obtain the subsequent table, yielding a space complexity of $O(\numrv^{\numos})$. Furthermore, obtaining each table requires updating $O(\numrv^{\numos})$ entries with a recurrence involving $O(\numos)$ operations, yielding a total of $O(\numos \numrv^{\numos})$ time complexity.

\subsection{Galgana and Shi's Insight}

The method given in Galgana and Shi, et. al. (2021) solves an $m$-distributions variant of \ijocdf($\mathcal{C}, \bm{p}$) in time and space complexity polynomial in both $\numrv$ and $\numos$. This result was an improvement over the best known previous method, as described in \cite{glueck2008algorithms}, as the latter runs in time exponential in $\numos$. The insight that allowed for this speedup 
distribution of $\bm{C}_i$. In combinatorial terms, the conditions $\event_{n, j+1:\numos}$ are independent of the specific configuration of the balls in the first $j$ bins, $(C_{n, 1},\ldots,C_{n, j})$, when given only the sum $C_{n, 1:j}$. Hence, we have:
\begin{lemma}[Galgana and Shi, et. al. (2021)]
    For all $j \in [d]$ and $k_i \in [\numrv^*]$ for all $i \in [j]$, it holds that:
    \begin{align*}
        \left(\event_{n,j+1:d} \perp\!\!\!\perp \bigcap_{i=1}^j \{ \ballcount_{n,i} = k_i \} \mid \{ \ballcount_{n,1:j} = \sum_{i=1}^j k_i \} \right)
        .
    \end{align*}

\end{lemma}
\noindent 
This Lemma suggests that the tables of probabilities of exact configurations maintained in Boncelet Jr.'s algorithm may be compressed, yet still retain all the information required to compute \ijocdf($\mathcal{C}, \bm{p}$). Indeed, our new method maintains significantly smaller tables at the cost of only a slightly more computationally expensive updating function. This improvement is accomplished by instead tracking the distribution of a \textit{function} of the $\bm{C}_{i}$'s.
\section{Our Solution}
\label{sec:solution}

In this section, we compress Boncelet Jr.'s probability tables to contain only information relevant to computing the joint cumulative distribution. The crux of our solution is that we can define transform $S : [\numrv^*]^\numos \to [\numrv^*]^\numos$ operating on ball count configurations that allows us to ignore the exact placement of ``superfluous'' balls. In accordance with our goal of efficiently computing \ijocdf($\mathcal{C}, \bm{p}$), this $S$ must have the following properties:
\begin{enumerate}
    \item The distribution of $S(\bm{C}_{n})$ is easier to maintain than that of $\bm{C}_n$.
    \item We must be able to compute $\mathbb{P}(\event_{n, 1:\numos})$ given only the distribution of $S(\bm{C}_{n})$.
\end{enumerate}
Motivating our choice of $S$ is the following observation: conditional on $C_{\numrv, 1:j} \geq c_j$, we can treat any additional balls that land in bin $j$ to be indistinguishable from those landing in bin $j+1$. That is, for $\bm{k}$ such that $k_{1:j} = \sum_{i=1}^j k_i > c_j$, we have:
$\prob ( \event_{\numrv, 1:\numos} \mid \bm{C}_{\numrv} = \bm{k} ) = \prob ( \event_{\numrv, 1:\numos} \mid \bm{C}_{\numrv} = \bm{k} + e_{j+1} - e_j )$.
%
%
Hence, in this case we can `spill' additional balls from bin $j$ to bin $j+1$ starting from $j = 1$ so long we maintain $\event_{\numrv, j}$. Analytically, this is equivalent to defining $S$ such that $S_j(\bm{k}) = \min(c_j, \max(0, k_{1:j}-c_{j-1}))$ where $S_j(\bm{k})$ denotes the $j$th entry of $S(\bm{k})$. This information compression scheme satisfies both criteria by: 1) capping the number of balls in bin $j$ at $c_j$ rather than $\numrv$ and; 2) still allows us to compute $\prob(\event_{n, 1:\numos})$ as the sum of all $\prob(\bm{C}_{n} = \bm{k})$ such that for all $j \in [\numos]$, $\{ S_{1:j}(\bm{k}) = \sum_{i=1}^{j} S_i(\bm{k}) \geq c_j\}$ holds. This latter property derives from balls only spilling to the right, hence $k_{1:j} \geq S_{1:j}(\bm{k})$, from which it follows that $\{S_{1:j}(\bm{k}) \geq c_j\} \implies \{k_{1:j} \geq c_j\} \implies \event_{\numrv, j}$. While this transformation $S$ possesses both our aforementioned desired qualities, there is still room for improvement. 

Under our current transform $S$, we guarantee that $S_j(\bm{k}) \leq c_j$. The distribution of $S(\bm{k})$ is easier to store and compute as the space that $S$ maps to---namely $\bigotimes_{j=1}^\numos [c_j^*]$---is smaller than the space of $\bm{k}$, which is $[\numrv^*]^{\numos}$. It would be even more efficient if we changed the spilling condition of $S$, which is to spill over when $C_{1:j} > c_j$, in such a way that guarantees that $\sum_{i=1}^j S_i(\bm{k}) \leq c_j$. This would further reduce the memory requirements to store the distribution of $S(\bm{k})$ whilst still maintaining the property that $\{S_{1:j}(\bm{k}) \geq c_j\} \implies \{k_{1:j} \geq c_j\} \implies \event_{\numrv, j}$.
In fact, we can define $S$ such that instead of spilling over when there are at least $c_j$ balls in the first $j$ bins, we instead spill over when there are at least $\delta_j = c_j - c_{j-1}$ balls in bin $I_j$. Analytically: 
\begin{definition}
    \label{def: spillover analytical}
    Let $\mathcal{C}, \bm{p}$ be as in \ijocdf($\mathcal{C}, \bm{p}$) and assume $c_0 = 0$. We define the \textit{spilling} transformation $S : [\numrv^*]^\numos \to \bigotimes_{j \in \numos} [\delta_j^*]$ as follows:
    \begin{align*}
        S_j (\bm{k}) 
        = 
        \min \,
        (
        \delta_j, \max_{j' \in [j]} \,
        (
        \sum_{i=j'}^{j} k_{i} - \sum_{i=j'}^{j-1} \delta_i 
        )
        )
        .
    \end{align*}
\end{definition} 
With this definition of $S$, we can show equivalence between the event that $S(\bm{C}_n) = (\delta_1,\ldots,\delta_{\numos})$ and bin conditions $\event_{n, 1:\numos}$.

\begin{theorem}
    For any $\bm{C}_n \in [\numrv]^\numos$ and $\bm{\delta} = (\delta_1,\ldots,\delta_{\numos})$, we have that $\event_{n, 1:\numos}$ holds if and only if $S(\bm{C}_n) = \bm{\delta}$. Hence $\mathbb{P}\left(S(\bm{C}_n)) = \bm{\delta}\right) = \mathbb{P}(\event_{n, 1:\numos})$.
\end{theorem}
\begin{proof}
    We first start with the forward direction. Assume that $\event_{n, 1:\numos}$ holds. By definition,
    \begin{align*}
        \event_{n, 1:\numos} 
        \leftrightarrow 
        \bigcap_{j=1}^\numos \{\ballcount_{n, 1:j} \geq c_j\}
        \leftrightarrow 
        \bigcap_{j=1}^\numos \{\ballcount_{n, 1:j} \geq \sum_{i=1}^j \delta_i\}
        .
    \end{align*}
    As $( \sum_{i=1}^j \ballcount_{n, i} - \delta_{i} )$ is non-negative,
    \begin{align*}
        \delta_j 
        \leq 
        ( \sum_{i=1}^j \ballcount_{n, i} - \delta_{i} ) + \delta_j 
        \leq 
        \max_{j' \in [j]} \, ( \ballcount_{n, j:j'} - \sum_{i=j'}^{j-1} \delta_i )
        .
    \end{align*}
    Following this, we have for all $j \in [\numos]$:
    \begin{align*}
        S_j(\bm{C}_n) 
        = 
        \min \, ( \delta_j, \max_{j' \in [j]} \, ( \ballcount_{n, j':j} - \sum_{i=j'}^{j-1} \delta_i ) ) 
        = \delta_j
        .
    \end{align*}
    Now we prove the backwards direction; assume that $S_j(\bm{C}_n) = \delta_j$ for all $j$. We will show that $\bigcap_{i=1}^j \{S_i(\bm{C}_n) = \delta_i\}$ implies $\event_{n, 1:j}$ by induction. Starting with the base case $j = 1$:
    \begin{align*}
        \delta_1 
        &= 
        S_1(\bm{C}_n) 
        = \
        \min \, ( \delta_1, \max_{j' \in [1]} ( \ballcount_{n, j':1} - \sum_{i=j'}^{0} \delta_i ) ) 
        \\
        &\geq 
        \min \, (\delta_1, \ballcount_{n, 1})
        .
    \end{align*}
    Thus, $\ballcount_{n, 1} \geq \delta_1 = c_1$ and the base case holds. For the recursive case:
    %
    \begin{align*}
        &\delta_j = S_j(\bm{C}_n) = \min \, ( \delta_j, \max_{j' \in [j]} \, ( \ballcount_{n, j':j} - \sum_{i=j'}^{j-1} \delta_i ) )
        \\
        &\implies 
        \delta_j \leq \max_{j' \in [j]} \, ( \ballcount_{n, j':j} - \sum_{i=j'}^{j-1} \delta_i )
        .
    \end{align*}
    Adding $\sum_{i=1}^{j-1} \delta_i$ to both sides, we have:
    \begin{align*}
        &\delta_j + \sum_{i=1}^{j-1} \delta_i 
        \leq 
        \max_{j' \in [j]} \, ( \ballcount_{n, j':j} + \sum_{i=1}^{j'-1} \delta_i ) 
        \\
        &= 
        \max_{j' \in [j]} \, ( \ballcount_{n, j':j} + c_{j'-1} )
        \\
        &= 
        \max_{j' \in [j]} \, ( \ballcount_{n, 1:j} - (\ballcount_{n, 1:j'-1} - c_{j'-1} ) )
        .
    \end{align*}
    However, by strong induction we know that $\ballcount_{n, 1:j'} \geq c_{j'}$ for all $j' < j$, thus the last term is always non-negative. Hence, $j' = 1$ corresponds to the maximum and as $\delta_j + \sum_{i=1}^{j-1} \delta_i = c_j$, we finish with:
    \begin{align*}
        c_j \leq 
        \max_{j' \in [j]} \, \left(\ballcount_{n, 1:j} - \left( \ballcount_{n, 1:j'-1} - c_{j'-1}\right)\right) = \ballcount_{n, 1:j}
        .
    \end{align*}
\end{proof}

Having established $\mathbb{P}\left(S(\bm{C}_n)) = \bm{\delta}\right) = \mathbb{P}(\event_{n, 1:\numos})$, all that remains is to actually compute the probability of the former. In particular, we need to be able to compute this without implicitly computing the distribution of $\bm{C}_n$, as this would amount to the work that Boncelet Jr.'s algorithm already does. 
One approach is instead of recursively maintaining the distribution of $\bm{C}_i$ like in Boncelet Jr.'s algorithm, we instead track the distribution of $S(\bm{C}_i)$. Using the spilling analogy of $S$ as motivation, the recurrence relation comes naturally. To help with this, we define an additional set-valued function 
$\sigma(j, \bm{\kappa}) = \{j\} \cup { j' < j: \bigcap_{i = j'}^{j-1} \{\kappa_i = \delta_i\} }$ 
to denote the set of bins such that when a ball is thrown into any of these bins given $S(\bm{k}) = \bm{\kappa}$, the ball will land in or spill over into or past bin $I_j$ under transformation $S$. Note that to avoid confusion we use $\bm{k} \in [\numrv^*]^{\numos}$ and $\bm{\kappa} \in \bigotimes_{j=1}^{\numos} [\delta_j^*]$). There are only two ways to obtain the event $S(\bm{C}_i) = \bm{\kappa}$:
\begin{enumerate}
    \item If $S(\bm{C}_{i-1}) = \bm{\kappa}$ and ball $i$ is thrown into any of the bins indexed in $\sigma(\numos+1, \bm{\kappa})$---ball $i$ lands or spills into the untracked $I_{\numos+1}$.
    \item If $S(\bm{C}_{i-1}) = \bm{\kappa} - e_j$ and ball $i$ is thrown into any of the bins indexed in $\sigma(j, \bm{\kappa})$---ball $i$ fell directly or spilled over into $I_j$.
\end{enumerate}
\noindent With this observation, and the law of total probability, we can formally state our key recurrence relation to obtain $\prob(S(\bm{C}_i) = \bm{\kappa})$ as a function of the distribution of $S(\bm{C}_{i-1})$.

\begin{theorem}
    \label{thm: spillover recurrence}
    Let $\mathcal{C}, \bm{p}$ be as in \ijocdf($\mathcal{C}, \bm{p}$) and define $S$ to be as in Definition~\ref{def: spillover analytical}. Then:
    \begin{align}
        &\prob(S(\bm{C}_i) = \bm{\kappa}) = 
        \sum_{j' \in \sigma(\numos+1, \bm{\kappa})} p_{i, j'} \prob(S(\bm{C}_{i-1}) = \bm{\kappa}) \nonumber
        \\
        &+ 
        \sum_{j=1}^\numos \prob(S(\bm{C}_{i-1}) = \bm{\kappa} - e_{j}) \sum_{j' \in \sigma(j, \bm{\kappa})} p_{i, j'} 
        .
    \end{align}
\end{theorem}

\subsection{Algorithm and Performance Analysis}

We begin by defining dynamic programming tables $T_i : \bigotimes_{j \in \numos} [\delta_j^*] \to [0, 1]$ for $i \in [\numrv^*]$ with the intention that with the recurrence relation from above, we have $T_i(\bm{\kappa}) = \prob(S(\bm{C}_i) = \bm{\kappa})$. We describe our procedure in Algorithm~\ref{alg:spillover dynamic program}.

\begin{algorithm}
\caption{Spilling Algorithm to solve \ijocdf($\mathcal{C}, \bm{p}$)}
\label{alg:spillover dynamic program}
\begin{algorithmic}[1]
\item[]
\Require{$\numrv, \numos \in \mathbb{N}, \numrv \geq \numos, \mathcal{C} = (c_1,\ldots,c_{\numos}) \in [\numrv]^{\numos}, \bm{p} \in [0, 1]^{\numrv \times (\numos+1)}$}
\Ensure{$\prob \left( \event_{\numrv, 1:\numos} \right)$}
  \State{$\delta_j = c_j - c_{j-1}, \forall j \in [\numos]$ for $c_{j-1} = 0$}
  \State{Define tables $T_i : \bigotimes_{j \in \numos} [\delta_j^*] \to [0, 1], \forall i \in [\numrv^*]$}
  \State {Initialize $T_0(\bm{\kappa})$ $\gets 1$ if $\bm{\kappa} = \bm{0}$, else 0 for $\bm{\kappa} \in \bigotimes_{j \in \numos} [\delta_j^*]$}
  \For{$i \gets 1$ to $\numrv$}
    \State{$\sigma(j, \bm{\kappa}) \doteq \{j\} \cup \Big\{j' < j: \bigcap_{i = j'}^{j-1} \{\kappa_i = \delta_i\} \Big\}, \forall j \in [\numos^+]$}
    \State {$T_i(\bm{\kappa}) = \sum_{j' \in \sigma(\numos+1, \bm{\kappa})} p_{i, j'} T_{i-1}(\bm{\kappa}) + \sum_{j=1}^\numos T_{i-1}(\bm{\kappa} - e_{j}) \sum_{j' \in \sigma(j, \bm{\kappa})} p_{i, j'}$}
  \EndFor
  \State \textbf{return} $\prob \left( \event_{\numrv, 1:\numos} \right) = T_\numrv(\bm{\delta})$ for $\bm{\delta} = (\delta_1,\ldots,\delta_{\numos})$ 
\end{algorithmic}
\end{algorithm}

While there are $\numrv+1$ tables of size 
$O ( \prod_{j \in \numos} (1 + \delta_j) )$, 
table $T_i$ may be deleted once $T_{i+1}$ is computed, hence the memory requirement is 
$O ( \prod_{j \in \numos} (1 + \delta_j) )$. 
This is a significant improvement over Boncelet Jr.'s algorithm, which requires $O(\numrv^{\numos})$ memory. Each entry of table $T_i$ for $i \in [\numrv]$ requires summing over $O(\numos)$ terms as per Theorem~\ref{thm: spillover recurrence}, which yields a time complexity of 
$O (\numos^2 \numrv \prod_{j \in \numos} (1 + \delta_j) )$.
The worst case improvement corresponds to when the order statistics given by $\mathcal{C}$ are evenly spaced, which yields $\delta_j \approx \frac{\numrv}{\numos}$ for all $j \in [\numos]$. Here, one can take the ratio of Boncelet Jr.'s time and space complexities versus that of Algorithm~\ref{alg:spillover dynamic program} to obtain a $O(\frac{d^{d-1}}{\numrv})$ and $O(\numos^{\numos})$ factor of improvement in time and space complexity respectively. 

One possible modification to our algorithm pre-computing and store the values of $\sum_{j \in \sigma(j', \bm{\kappa})} p_{i, j}$, which saves a factor of $\numos$ computations, at a cost of a factor of $\numos$ memory. Pre-computing these probability sums yields a total time complexity of 
$O ( \numos \numrv \prod_{j \in \numos} (1 + \delta_j) )$ 
and space complexity of 
$O ( \numos \prod_{j \in \numos} (1 + \delta_j) )$. 
We can further improve the algorithm's performance by a constant factor by ignoring intermediary entries $T_i(\bm{\kappa})$ where the event that $S(\bm{C}_i) = \bm{\kappa}$ guarantees that $S(\bm{C}^\numrv)$ cannot possibly be equal to $\bm{\delta}$. This is the case if $\sum_{j=1}^\numos (\kappa_j - \delta_j) > \numrv - i$, as there are not enough balls to fill the remaining bins. With either algorithm modification, the spilling algorithm is efficient in cases where the desired order statistic indices are close together and where $\numos$ is relatively small.

\subsection{Experiments}

To verify our algorithm's performance guarantees, and directly compare with Boncelet Jr.'s algorithm, we conduct two experiments which varies the values of $\numrv, \numos, \mathcal{C}$. We use the version of the algorithm without pre-computing the probability sums but disregarding irrelevant intermediary entries. The experiments were run a computer with an Intel Xeon E3-1240V5 CPU with 15.6 GB of memory.
\begin{figure}[t]
    \centering
    \includegraphics[scale=0.2]{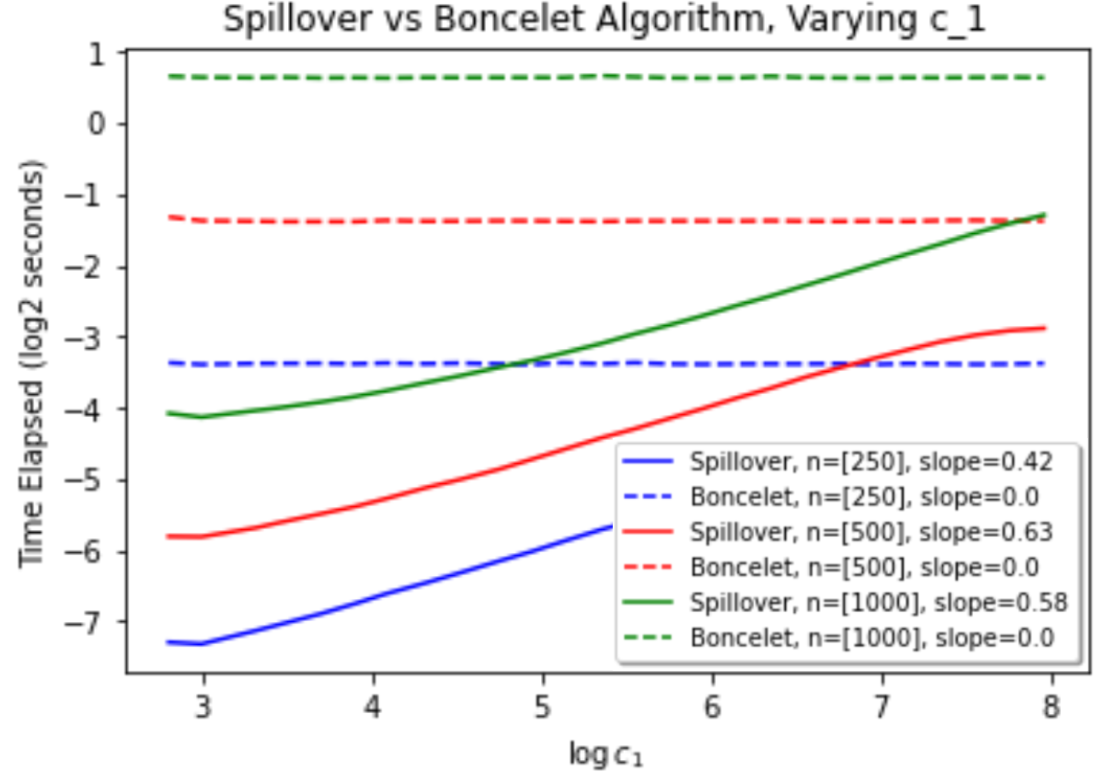}
    \caption{\small For our first experiment, we let $\numrv \in \{250, 500, 1000\}, c_1 \in [100], \numos = 1$. The slope of the spillover run-times indicate a worst case linear function of $c_1 = \delta_1$. The equally sized parallel shifts in times when varying $n$ by a multiplicative factor suggests polynomial dependence in $n$.}
\label{fig:experiment_1}
\end{figure}

\begin{table}
\small
%
\begin{tabular}{@{}lllllll@{}}
\toprule
$n$ & $6$ & $12$ & $18$ & $24$ & $30$ \\
\midrule
$d = 1$ & 6.9E-5 & 1.1E-4 & 2.0E-4 & 1.7E-4 & 1.8E-4  \\
$d = 2$ & 8.5E-4 & 1.6E-4 & 2.3E-4 & 2.7E-4 & 3.7E-4  \\
$d = 3$ & 1.8E-4 & 3.1E-4 & 4.4E-4 & 5.8E-4 & 7.6E-4  \\
$d = 4$ & 3.4E-4 & 6.4E-4 & 9.5E-4 & 1.4E-3 & 1.6E-3  \\
$d = 5$ & 7.2E-4 & 1.4E-3 & 2.1E-3 & 2.8E-3 & 3.7E-3  \\
$d = 6$ & 1.5E-3 & 3.0E-3 & 5.6E-3 & 6.5E-3 & 7.9E-3  \\ \bottomrule
\end{tabular}
\caption{\small
\textbf{Spillover Algorithm Elapsed Time}. We ran an experiment from Galgana and Shi, et. al. (2021) with $\numrv \in \{6, 12, 18, 24, 30\}$ and $\mathcal{C} \in \{[\numos]\}$ for $\numos \in [6]$.}
\label{Spillover-times}
\end{table}

%
%
\begin{table}
\small
\begin{tabular}{@{}lllllll@{}}
\toprule
$n$ & $6$ & $12$ & $18$ & $24$ & $30$ \\
\midrule
$d=1$ & 1.1E-4 & 4.0E-4 & 5.8E-4 & 1.1E-3 & 1.6E-3 \\ 
$d=2$ & 6.4E-4 & 2.8E-3 & 7.3E-3 & 1.7E-2 & 3.0E-2 \\ 
$d=3$ & 1.3E-3 & 1.4E-2 & 6.1E-2 & 1.7E-1 & 4.1E-1 \\ 
$d=4$ & 3.8E-3 & 6.5E-2 & 3.8E-1 & 1.4E0 & 4.0E0 \\ 
$d=5$ & 9.4E-3 & 2.4E-1 & 2.0E0 & 9.4E0 & 3.2E1 \\ 
$d=6$ & 2.1E-2 & 8.4E-1 & 8.8E0 & 5.3E1 & 2.2E2
\\\bottomrule
\end{tabular}
\caption{\small \textbf{Boncelet Jr.'s Algorithm Elapsed Time}. Our algorithm is significantly faster than Boncelet Jr.'s especially for larger $\numrv$ or $\numos$.}
\label{Boncelet-times}
\end{table}

\section{Extensions}
\label{sec:extensions}

One of the key advantages to using Boncelet Jr.'s algorithm over its competitors is its flexibility to handle dependent random variables. 
Specifically, Boncelet Jr.'s algorithm allows for dependency between random variables, albeit with additional memory and computational cost. 
This flexibility can be especially useful for various stochastic processes such as Markov processes or more general autoregressive models 
A common application of computing order statistic distributions for dependent random variables is estimating the maximum length of a queue. The generalization to dependent random variables requires some changes to the information stored in the dynamic programming tables, as we will need to understand the conditional distribution of balls yet to be thrown given the placement of the balls already thrown. Of course, the task of tracking the exact placement of each ball was one that we wanted to avoid in the first place. In this section, we show that, in cases of sparse random variable dependency structures, then we can limit our memory of exact placement to a small subset of the balls.

\subsection{Preliminaries}

Given $\numrv$ random variables and their joint distribution $\bm{F}$, we can construct the corresponding Markov random field (MRF). In order to minimize the number of random variables tracked in our algorithm, we utilize the local Markov property and the structure of the conditional dependencies in the MRF. We adopt most of the combinatorial notation from the previous section, however we will use an additional partitioning of each bin to simulate exact placement of each random variable. Letting node $i$ of the MRF represent variable $X_i$: \footnote{We can instead let node $i$ represent a different variable $v_i$, which creates a different visitation scheme $\bm{v}$, neighbor sets, and boundary sets. Consequently, the choice of visitation scheme also impacts the algorithm's space and time complexity, and in general, computing the optimal re-ordering is NP-hard. For simplicity, we stick to the original ordering.}
\begin{enumerate}
    \item Define micro-bins and micro-bounds $I_{j, h}$ and $x_{j, h}$ for $j \in [\numos^+]$ and $h \in [H]$, for granularity factor $H \in \mathbb{N}$. Here, $x_j = x_{j, 0} \leq \ldots, \leq x_{j, H} = x_{j+1}$ and $I_{j, h} = (x_{j, h-1}, x_{j, h}]$.
    \item Define the vector $\mathcal{N}_i \in ([\numos^+] \times [H])^{n(i)}$ to denote the micro-bin locations of the neighboring set of $i$---the lower-indexed neighbors of ball $i$. By the local Markov property, $i$ is conditionally independent of $\{1,\ldots,i-1\}$ given $\mathcal{N}_i$. Here, $n(i)$ denotes the number of lower-indexed neighbors of $i$.
    \item Define the vector $\mathcal{B}_i \in ([\numos^+] \times [H])^{b(i)}$ to denote the micro-bin locations of the boundary set of $i$---the neighbors of balls $i,\ldots,\numrv$ among the first $i-1$ balls. Again by the local Markov property, balls $i$ through $\numrv$ are conditionally independent of $\{1,\ldots,i-1\}$ given $\mathcal{B}_i$. Here, $b(i)$ denotes the number of neighbors of balls $i$ through $\numrv$ among the first $i-1$ balls.
    \item Define the vector-valued transform $N_i : ([\numos^+] \times [H])^{b(i)} \to ([\numos^+] \times [H])^{n(i)}$ to select only the entries of $\mathcal{B}_i \in ([\numos^+] \times [H])^{b(i)}$ corresponding to the neighbors of $i$, namely $\mathcal{N}_i$. That is, $N_i(\mathcal{B}_i) = \mathcal{N}_i$.
    \item Define $\bm{p}_i(\cdot) \doteq \bm{p}_i(\cdot \mid N_i(\mathcal{B}_i) = \bm{j}_{\mathcal{N}}) \doteq \bm{p}_i(\cdot \mid \mathcal{N}_i = \bm{j}_{\mathcal{N}})$ to be the conditional distribution of ball $i$ given $\{\mathcal{N}_i = \bm{j}_{\mathcal{N}}\}$. Here, $\bm{j}_{\mathcal{N}} \in ([\numos^+] \times [H])^{n(i)}$. In practice, $p_i(\cdot)$ can be estimated through MCMC-sampling methods or by numerical integration and marginalizing out its higher-indexed neighbors. If the random variables are discrete, it may be computed exactly. 
\end{enumerate}

With our new notation, we can formally define the dependent version of \ijocdf({$\mathcal{C}, \bm{p}$}).

\begin{definition}
Assuming random variables $X_1,\ldots,X_{\numrv}$, we let $\numrv, \numos, \bm{\rvv}, \mathcal{C}, \bm{F}$ be as in Definition~\eqref{def:jocdf_cts}. We define $\bm{p} = (p_{i, j, h}(\bm{j}_{\mathcal{N}}))_{i \in [\numrv], j \in [\numos^+], h \in [H], \bm{j}_{\mathcal{N}} \in ([\numos^+] \times [H])^{n(i)}} = (F_i(\rvv_j \mid \mathcal{N}_i = \bm{j}_{\mathcal{N}}) - F_i(\rvv_{j-1} \mid \mathcal{N}_i = \bm{j}_{\mathcal{N}}))_{i \in [\numrv], j \in [\numos], h \in [H], \bm{j}_{\mathcal{N}} \in ([\numos^+] \times [H])^{n(i)}}$. That is, $p_{i, j, h}(\bm{j}_{\mathcal{N}})$ is the probability that ball $i$ falls into micro-bin $I_{j, h}$ given $\mathcal{N}_i = \bm{j}_{\mathcal{N}}$. We define the combinatorial problem \djocdf $\left(\mathcal{C}, \bm{p}\right)$ as computing the probability $\prob \left( \event_{\numrv, 1:\numos} \right)$.
\end{definition}


At first, it seems there are only three additions to Algorithm~\ref{alg:spillover dynamic program} required to solve Equation~\eqref{eq:jocdf_cts} with dependent random variables. First, we must further sub-divide each bin. Second, we need our algorithm to store $\mathcal{N}_i$ in micro-bin space by iteration $i$. Third, we need to change our definition of $\bm{p}_i$ to reflect the dependency between the current random variable and those in $\mathcal{N}_i$. In particular, instead of having $\bm{p}_i: [\numos] \to [0, 1]$ as the singular distribution of the current ball over all bins, we can instead have a distribution $\bm{p}_i(\cdot) = \bm{p}_i(\cdot \mid \mathcal{N}_i = \bm{j}_{\mathcal{N}})$. However, these three changes are insufficient as the algorithm must also track the location of balls that will be needed to compute the conditional distribution of not only ball $i$ but also $i+1,\ldots,\numrv$. More specifically, our algorithm must not ``forget'' the location of the balls in $\mathcal{B}_i$ by storing the configuration in memory at each time step $i$.

\subsection{Updated Recurrence}

The updated recurrence relation is similar to Theorem~\eqref{thm: spillover recurrence}, except now we also need to take into account $\mathcal{B}_{i}$ and how they both impact the distribution of ball $i$ and also the placement of balls $i+1$ through $\numrv$. Moreover, we will also benefit from defining set-valued function $\psi_i$ that takes argument $\bm{j}_{\mathcal{B}} \in ([\numos^+] \times [H])^{b(i+1)}, \bm{\kappa} \in \bigotimes_{j \in \numos} [\delta_j^*]$ and returns the set of $\bm{j}_{\mathcal{B}}' \in ([\numos^+] \times [H])^{b(i)}$ such that $\{\mathcal{B}_i = \bm{j}_{\mathcal{B}}'\}$, $\{\mathcal{B}_{i+1} = \bm{j}_{\mathcal{B}}\}$, and $\{S(\bm{C}_i) = \bm{\kappa}\}$ are non-disjoint events---the ball configurations across the three events are consistent. Lastly, we must make an important distinction between when $i$ is in the boundary set of $i+1$. As such, we define a piece-wise function $\gamma_i$ that takes as arguments $\bm{j}_{\mathcal{B}} \in ([\numos^+] \times [H])^{b(i+1)}, \bm{j}_{\mathcal{B}}' \in ([\numos^+] \times [H])^{b(i)}, \bm{\kappa} \in \bigotimes_{j \in \numos} [\delta_j^*]$. If $i$ is in the boundary set of $i+1$, then let $j, h$ denote the last entry of $\bm{j}_{\mathcal{B}}$ and set $\gamma_i(\bm{j}_{\mathcal{B}}, \bm{j}_{\mathcal{B}}', \bm{\kappa}) = p_{i, j, h}(N_i(\mathcal{B}_i) = \bm{j}_{\mathcal{B}}')$. Otherwise, we set it to $\sum_{j' \in \sigma(j, \bm{\kappa})} \sum_{h=1}^H p_{i, j', h}(N_i(\mathcal{B}_i) = \bm{j}_{\mathcal{B}}')$. As ball $i$ is conditionally independent of the event $S(\bm{C}_i) = \bm{\kappa}_i$ given $\{N_i(\mathcal{B}_i) = \bm{j}_{\mathcal{N}}\}$ or $\{\mathcal{N}_i = \bm{j}_{\mathcal{N}}\}$, the following recurrence relation follows from the law of total probability, much like in Equation~\eqref{thm: spillover recurrence}:

\begin{theorem}
    \label{thm: spillover recurrence dependent not in}
    Let $\mathcal{C}, \bm{p}$ be as in \djocdf($\mathcal{C}, \bm{p}$). Define $S$ to be as in Definition~\ref{def: spillover analytical} and $\mathcal{N}_i, \mathcal{B}_i, \psi_i, \gamma_i$ to be as in the previous section. Then:
    %
    Then:
    \begin{align}
        &\prob \left(S(\bm{C}_i) = \bm{\kappa}, \mathcal{B}_{i+1} = \bm{j}_{\mathcal{B}} \right) 
        \\
        \nonumber
        =&
        \sum_{j' \in \sigma(\numos + 1, \bm{\kappa})}
        \sum_{\bm{j}_{\mathcal{B}}' \in \psi_i \left( \bm{j}_{\mathcal{B}}, \bm{\kappa} \right)}
        \prob \left( S \left( \ballcount^{i-1} \right) = \bm{\kappa}, \mathcal{B}_{i} = \bm{j}_{\mathcal{B}}' \right) 
        \\
        \nonumber
        &\qquad \gamma_i(\bm{j}_{\mathcal{B}}, \bm{j}_{\mathcal{B}}', \bm{\kappa})
        \\
        \nonumber
        &+ 
        \sum_{j=1}^{\numos} 
        \sum_{\bm{j}_{\mathcal{B}}' \in \psi_i \left( \bm{j}_{\mathcal{B}},
        \bm{\kappa} \right)} 
        \prob \left( S \left( \ballcount^{i-1} \right) = \bm{\kappa} - e_{j}, \mathcal{B}_{i} = \bm{j}_{\mathcal{B}}' \right) 
        \\
        \nonumber
        &\gamma_i(\bm{j}_{\mathcal{B}}, \bm{j}_{\mathcal{B}}', \bm{\kappa})
        .
    \end{align}
\end{theorem}

\subsection{Algorithm and Performance Analysis}

We begin by defining dynamic programming tables $T_i : \bigotimes_{j \in \numos} [\delta_j^*] \times ([\numos^+] \times [H])^{b(i)} \to [0, 1]$ for $i \in [\numrv]$ with the intention that with the recurrence relation from above, $T_i(\bm{\kappa}, \bm{j}_{\mathcal{B}}) = \prob(S(\bm{C}_i) = \bm{\kappa}, \mathcal{B}_{i+1} = \bm{j}_{\mathcal{B}})$. We describe our procedure in Algorithm~\ref{alg:spillover dynamic program dependent}.

\begin{algorithm}
\caption{Spilling Algorithm to solve \djocdf($\mathcal{C}, \bm{p}$)}
\label{alg:spillover dynamic program dependent}
\begin{algorithmic}[1]
\item[]
\Require{$\numrv, \numos, H \in \mathbb{N}, \numrv \geq \numos, \mathcal{C} \in [\numrv]^{\numos}, \bm{p} \in [0, 1]^{\numrv \times (\numos+1) \times H \times (H(\numos+1))^{n(i)}}$}
\Ensure{$\prob \left( \event_{\numrv, 1:\numos} \right)$}
  \State{$\delta_j = c_j - c_{j-1}, \forall j \in [\numos]$ for $c_{j-1} = 0$}
  \State{Define tables $T_i : \bigotimes_{j \in \numos} [\delta_j^*] \to [0, 1], \forall i \in [\numrv^*]$}  
  \State {Initialize $T_0(\bm{\kappa}, \emptyset)$ $\gets 1$ if $\bm{\kappa} = \bm{0}$, else 0 for $\bm{\kappa} \in \bigotimes_{j \in \numos} [\delta_j^*]$}
  \For{$i \gets 1$ to $\numrv$}
    \For{$\bm{j} \in [H(\numos+1)]^{b(i+1)}$}
        \State{$\alpha 
        = 
        \sum_{j' \in \sigma \left( \numos + 1, \bm{\kappa} \right)}
        \sum_{\bm{j}_{\mathcal{B}}' \in \psi_i \left( \bm{j}_{\mathcal{B}}, \bm{\kappa} \right) } 
        T_i \left( \bm{\kappa}, \bm{j}_{\mathcal{B}}' \right) 
        \allowbreak
        \gamma_i(\bm{j}_{\mathcal{B}}, \bm{j}_{\mathcal{B}}', \bm{\kappa}) 
        $}
        \State{$\beta 
        =
        \sum_{j=1}^{\numos} 
        \sum_{\bm{j}_{\mathcal{B}}' \in \psi_i \left( \bm{j}_{\mathcal{B}}, \bm{\kappa} \right)} 
        T_i \left( \bm{\kappa} - e_{j}, \bm{j}_{\mathcal{B}}' \right) 
        \allowbreak
        \gamma_i(\bm{j}_{\mathcal{B}}, \bm{j}_{\mathcal{B}}', \bm{\kappa})
        $}
        \State{$T_i \left( \bm{\kappa}, \bm{j}_{\mathcal{B}} \right) = 
        \alpha + \beta$}
    \EndFor
  \EndFor
  \State \textbf{return} $\prob \left( \event_{\numrv, 1:\numos} \right) = T_\numrv(\bm{\delta}, \emptyset)$ for $\bm{\delta} = (\delta_1,\ldots,\delta_{\numos})$
\end{algorithmic}
\end{algorithm}

Just like in Algorithm~\ref{alg:spillover dynamic program},
table $T_i$ is disposable once table $T_{i+1}$ is computed. Hence, the space complexity is bounded by size of the largest table, which is of size 
$O ( (H\numos)^{b(*)} \prod_{j \in \numos} (1 + \delta_j) )$,
where $b(*) = \max_{i \in [\numrv]} b(i)$.
We note that $b(*)$ is dependent on the graph visitation scheme, which we assume to be simply $1,\ldots,\numrv$ wlog. Intuitively, the memory required grows in the size of the largest boundary set. 
With regards to the time complexity, we have that the updating step to compute $T_i(\bm{\kappa}, \bm{j}_{\mathcal{B}})$ requires a summation over $O(\numos)$ possible values of $j'$, $O(|\psi_i(\bm{j}_{\mathcal{B}}, \bm{\kappa})|)$ values of $\bm{j}_{\mathcal{B}}'$, and another $O(\numos)$ possible values of $j$. Here, the quantity $|\psi_i(\bm{j}_{\mathcal{B}}, \bm{\kappa})|$ is difficult to measure, so we upper bound it simply by $O((H\numos)^{\Delta_i})$, where $\Delta_i$ denotes the number of variables included in the boundary set of $i$ but not of $i+1$. 
Notice that  $\Delta_i \leq b(i) \leq b(*)$. 
With this, since there are 
$O ( (H\numos)^{b(*)} \prod_{j \in \numos} (1 + \delta_j) )$ 
table entries in $T_i$ to compute for each $i \in [\numrv]$, we have a total time complexity of 
$O ( \numrv(H\numos)^{1+2b(*)} \prod_{j \in \numos} (1 + \delta_j) )$.

\subsection{Experiments}

We experimentally compare our algorithm's output against the Monte Carlo approximation of the joint cdf of Markov random variables using the same computer system described earlier. In particular, we define the random walk $X_1,\ldots,X_\numrv$ such that $X_0 = 0$ and $X_{i+1} \in \{X_i - 1, X_i, X_i + 1\}$ with mean $\mu$, for $i \in [\numrv]$, $\numrv \in [30,...,365]$. 
We computed the probability $\prob(X_{(\lfloor\frac{90n}{100}\rfloor)} \leq 1, X_{(\lfloor\frac{95n}{100}\rfloor)} \leq 2, X_{(\lfloor\frac{99n}{100}\rfloor)} \leq 3)$. 
We note that with sufficiently large $\numrv$, the biased random walk almost surely converges to a geometric Brownian motion, which gives our algorithm many practical applications. One such application is managing downside risk in financial portfolios. More specifically, we define $X_i$ to be the aggregate \textit{loss} of a portfolio after $i$ time steps.  
See Figure~\ref{fig:experiment_3}.
With $\mathcal{C} = \{\lfloor\frac{90n}{100}\rfloor, \lfloor\frac{95n}{100}\rfloor, \lfloor\frac{99n}{100}\rfloor\}$, $\bm{x} = (3, 5, 10)$, \djocdf($\mathcal{C}, \bm{p}$) models the probability of exceeding a maximum downside risk constraint corresponding to the 90th percentile, 95th percentile, and 99th percentile losses in a portfolio.

\begin{figure}[t]
    \centering
    \includegraphics[scale=0.19]{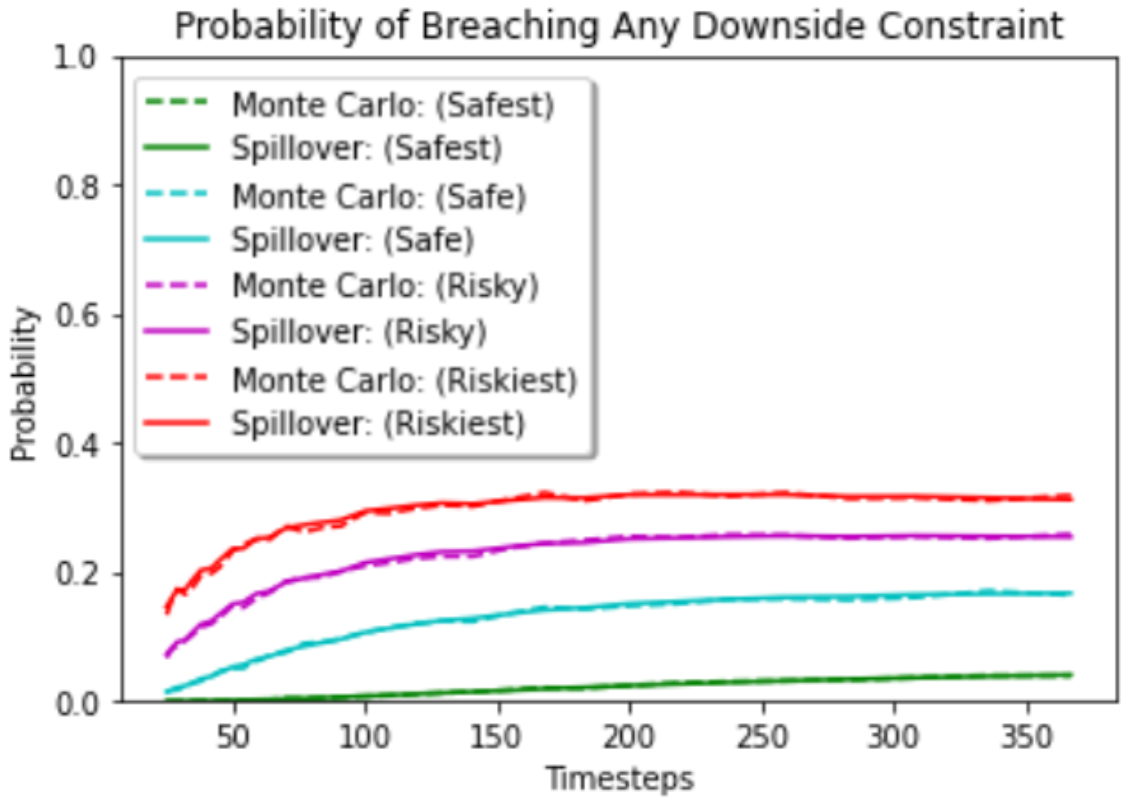}
    \caption{\small For our third experiment, we define 4 portfolios with expected daily returns of 1 (safest), 2 (safe), 3 (risky), and 4 (riskiest) basis points respectively and annualized Sharpe ratio of 1. We plot the probability of violating any of the downside risk constraints over varying time horizons. Observe that the Monte Carlo (with 10000 trials) track our algorithm's output very closely.}
\label{fig:experiment_3}
\end{figure}

\section{Conclusion}
\label{sec:conclusion}

In this paper, we describe the computation of the joint cdf of $d$ order statistics of $\numrv$ random variables; in particular, we focus on the method provided by \cite{boncelet1987algorithms}, and show how the computational complexity of Boncelet Jr.'s algorithm can be improved by compressing dynamic programming tables to store only information pertinent to computing the cdf rather than the entire joint distribution of the ball configurations. We extend our procedure to handle dependent random variables and show how this affects the space and time complexity. In the worst case of evenly spread $\numos$ order statistics, our algorithm improves the time over Boncelet Jr.'s by a factor of $O(\frac{\numos^{\numos-1}}{\numrv})$, and space complexity by $O( \numos^{\numos} )$.


\begin{thebibliography}{1}

\bibitem{balakrishnan2007permanents}
N. Balakrishnan, {\em Permanents, order statistics, outliers, and robustness}, Revista matem{\'a}tica complutense, 20 (2007), pp.~7--107.

\bibitem{benjamini1995controlling}
Y. Benjamini and Y. Hochberg, {\em Controlling the false discovery rate: a practical and powerful approach to multiple testing}, Journal of the Royal statistical society: series B (Methodological), 57 (1995), pp.~289--300.

\bibitem{shorack1996empirical}
G. Shorack and J. Wellner, {\em Empirical Processes with Applications to Statistics}, Society for Industrial and Applied Mathematics (1996)

\bibitem{varian}
H. L. Varian, {\em Position Auctions},  International Journal of Industrial Organization (2007), pp.~1163-1178

\bibitem{jones1995queues}
L. Jones and L. Larson, {\em Efficient Computation of Probabilities of Events Described by Order Statistics and Applications to Queue Inference}, Operations Research Center Working Paper (1994), pp.~289-294

\bibitem{yang2011order}
H. Yang and M. Alouini, {\em Order statistics in wireless communications: diversity, adaptation, and scheduling in MIMO and OFDM systems}

\bibitem{boncelet1987algorithms}
C. G. Boncelet Jr., {\em Algorithms to computer order statistic distributions}, SIAM Journal on Scientific and Statistical Computing, 8 (1987), pp.~868--876.

\bibitem{glueck2008algorithms}
D. H. Glueck, A. Karimpour-Fard, J. Mandel, L. Hunter, and K. E. Muller, {\em Fast computation by block permanents of cumulative distribution functions of order statistics from several populations}, Communications in Statistics: theory and methods, 37 (2008), pp.~2815--2824.
    
\bibitem{schroeder}
J. von Schroeder and T. Dickhaus, {\em Efficient calculation of the joint distribution of order statistics}, Computational Statistics and Data Analysis (2020)

\bibitem{galgana2021dynamic}
R. Galgana, C. Shi, A. Greenwald, and O. Takehiro, {\em A Dynamic Program for Computing the Joint Cumulative Distribution Function of Order Statistics}, SIAM Conference on Applied and Discrete Algorithms (2021).


\end{thebibliography}

\begin{thebibliography}{}
\setlength{\itemindent}{-\leftmargin}
\makeatletter\renewcommand{\@biblabel}[1]{}\makeatother

@inproceedings{jones1995queues,
  title={Efficient Computation of Probabilities of Events Described by Order Statistics and Applications to Queue Inference},
  author={Jones, Lee K. and Larson, Richard C},
  booktitle={Operations Research Center Working Paper},
  pages={289-294},
  year={1994}
}

@inproceedings{yang2011order,
  title={Order statistics in wireless communications: diversity, adaptation, and scheduling in MIMO and OFDM systems},
  author={H. Yang and M. Alouini},
  year={2007}
}

@inproceedings{varian,
  title={Position Auctions},
  author={H. L. Varian},
  booktitle={International Journal of Industrial Organization},
  pages={1163--1178},
  year={2007}
}

@inproceedings{schroeder,
  title={Efficient calculation of the joint distribution of order statistics},
  author={J. von Schroeder and T. Dickhaus},
  booktitle={Computational Statistics and Data Analysis},
  year={2020}
}


\end{thebibliography}
\end{document}